\documentclass[letterpaper, 10 pt, conference]{ieeeconf}  

\IEEEoverridecommandlockouts                              

\usepackage{graphics} 
\usepackage{epsfig} 
\usepackage{mathptmx} 
\usepackage{times} 
\usepackage{amsmath} 
\usepackage{amssymb}  
\usepackage{mathtools}
\usepackage{xcolor}
\usepackage{algorithm}
\usepackage{algorithmic}
\usepackage{tikz,tikz-network}
\usepackage{nccmath}
\usepackage{subcaption}

\makeatletter
\DeclareRobustCommand\widecheck[1]{{\mathpalette\@widecheck{#1}}}
\def\@widecheck#1#2{%
    \setbox\z@\hbox{\m@th$#1#2$}%
    \setbox\tw@\hbox{\m@th$#1%
       \widehat{%
          \vrule\@width\z@\@height\ht\z@
          \vrule\@height\z@\@width\wd\z@}$}%
    \dp\tw@-\ht\z@
    \@tempdima\ht\z@ \advance\@tempdima2\ht\tw@ \divide\@tempdima\thr@@
    \setbox\tw@\hbox{%
       \raise\@tempdima\hbox{\scalebox{1}[-1]{\lower\@tempdima\box
\tw@}}}%
    {\ooalign{\box\tw@ \cr \box\z@}}}
\makeatother

\newcounter{definition}
{\medskip}
\newenvironment{definition*}{\refstepcounter{definition}\par\medskip
\noindent 
\textbf{Definition \thedefinition.} \em \rmfamily}
{\medskip}

\newcounter{theorem}
\newenvironment{theorem}[1]{\refstepcounter{theorem}\par\medskip
\noindent 
\textbf{Theorem \thetheorem~(#1)} \em \rmfamily}
{\medskip}

\newcounter{assumption}
\newenvironment{assumption}{\refstepcounter{assumption}\par\medskip
\noindent 
\textbf{Assumption \theassumption.} \em \rmfamily}
{\medskip}

\newcounter{proposition}
\newenvironment{proposition}{\refstepcounter{proposition}\par\medskip
\noindent 
\textbf{Proposition \theproposition.} \em \rmfamily}
{\medskip}

\newenvironment{proposition*}[1]{\refstepcounter{proposition}\par\medskip
\noindent 
\textbf{Proposition \theproposition~(#1)} \em \rmfamily}
{\medskip}

\newcounter{lemma}
{\medskip}

\newenvironment{lemma*}[1]{\refstepcounter{lemma}\par\medskip
\noindent 
\textbf{Lemma \thelemma~(#1)} \em \rmfamily}
{\medskip}

\newcounter{problem}
{\medskip}

\newenvironment{problem*}[1]{\refstepcounter{problem}\par\medskip
\noindent 
\textbf{Problem \theproblem~(#1)} \em \rmfamily}
{\medskip}

\newcounter{remark}
\newenvironment{remark}{\refstepcounter{remark}\par\medskip
\noindent 
\textit{Remark \theremark.} \em \rmfamily}
{\medskip}

\DeclareMathOperator*{\argmax}{arg\,max} 

\newcommand{\be}{\begin{equation}}
\newcommand{\ee}{\end{equation}}
\newcommand{\nn}{\nonumber}
\newcommand{\br}{\mathrm{br}}
\newcommand{\E}{\mathrm{E}}

\newcommand{\uG}{\overline{G}}
\newcommand{\uE}{\overline{E}}
\newcommand{\cq}{\widecheck{q}}
\newcommand{\cu}{\widecheck{u}}
\newcommand{\hq}{\widehat{q}}

\newcommand{\aalpha}{\overline{\alpha}}
\newcommand{\QRE}{\mathrm{QRE}}

\newcommand{\pq}{\pmb{q}}
\newcommand{\ppi}{\pmb{\pi}}
\newcommand{\phq}{\pmb{\hq}}
\newcommand{\pcq}{\pmb{\cq}}

\definecolor{ferrarired}{rgb}{1.0, 0.11, 0.0}

\title{\LARGE \bf
Generalized Individual Q-learning for Polymatrix Games with \\Partial Observations
}

\author{Ahmed Said Donmez and Muhammed O. Sayin
\thanks{*This work was supported by The Scientific and Technological Research Council of T\"{u}rkiye (TUBITAK) BIDEB 2232-B International Fellowship for Early Stage Researchers under Grant Number 121C124.}
\thanks{A. S. Donmez and M. O. Sayin are with the Department of Electrical and Electronics Engineering,
        Bilkent University, Ankara, T\"{u}rkiye.
        Email: {\tt\small said.donmez@bilkent.edu.tr, sayin@ee.bilkent.edu.tr}}%
}

\begin{document}

\maketitle
\thispagestyle{empty}
\pagestyle{empty}

\begin{abstract}
This paper addresses the challenge of limited observations in non-cooperative multi-agent systems where agents can have partial access to other agents' actions. We present the generalized individual Q-learning dynamics that combine belief-based and payoff-based learning for the networked interconnections of more than two self-interested agents. This approach leverages access to opponents' actions whenever possible, demonstrably achieving a faster (guaranteed) convergence to quantal response equilibrium in multi-agent zero-sum and potential polymatrix games. Notably, the dynamics reduce to the well-studied smoothed fictitious play and individual Q-learning under full and no access to opponent actions, respectively. We further quantify the improvement in convergence rate due to observing opponents' actions through numerical simulations.
\end{abstract}

\section{INTRODUCTION}

Limited access to other agents' actions is a significant challenge for multi-agent systems, particularly for networked interconnections \cite{ref:Yuksel13}. In the absence of observations, the literature often focuses on the extreme case of no observation, i.e., payoff-based dynamics \cite{ref:Leslie05,ref:Marden09}. However, payoff-based dynamics ignoring partial observations available can suffer from slow convergence to equilibrium though asymptotic convergence guarantees might still be possible.


This paper addresses the limited access to opponent actions by proposing novel learning dynamics that combine belief-based and payoff-based learning for self-interested agents in non-cooperative environments. To this end, we focus on polymatrix games \cite{ref:Cai16}, where agents have network separable \textit{pairwise} interconnections while they can only observe the actions of the agents in proximity, as illustrated in Fig. \ref{fig:layers} over two separate layers.
We present the \textit{generalized} individual Q-learning dynamics, in which agents leverage opponents' actions that they can observe so that they can accelerate their learning process. The dynamics reduce to the well-studied smoothed fictitious play \cite{ref:Fudenberg09} and individual Q-learning \cite{ref:Leslie05}, in scenarios with access to the actions of every or none of the other agents, respectively.

\begin{figure}[t!]
\centering
\begin{tikzpicture}[multilayer=3d,scale=0.8]

\SetEdgeStyle[LineWidth=0.5]
\Vertex[x=-0.5,y=0.5,NoLabel,size=.1,layer=1]{A1}
\Vertex[x=0,y=1.3,NoLabel,size=.1,layer=1]{B1}
\Vertex[x=0.5,y=-0.5,NoLabel,size=.1,layer=1]{C1}
\Vertex[x=1.5,y=1.3,NoLabel,size=.1,layer=1]{D1}
\Vertex[x=1.5,y=0.5,NoLabel,size=.1,layer=1]{E1}
\Vertex[x=1,y=0.3,NoLabel,size=.1,layer=1]{F1}
\Vertex[x=2.5,y=-.3,NoLabel,size=.1,layer=1]{G1}
\Edge[Direct](A1)(B1)
\Edge[Direct](B1)(A1)
\Edge[Direct](A1)(C1)
\Edge[Direct](C1)(F1)
\Edge[Direct](F1)(C1)
\Edge[Direct](F1)(E1)
\Edge[Direct](E1)(G1)
\Edge[Direct](D1)(E1)
\Vertex[x=-0.5,y=0.5,NoLabel,size=.1,layer=2]{A2}
\Vertex[x=0,y=1.3,NoLabel,size=.1,layer=2]{B2}
\Vertex[x=0.5,y=-0.5,NoLabel,size=.1,layer=2]{C2}
\Vertex[x=1.5,y=1.3,NoLabel,size=.1,layer=2]{D2}
\Vertex[x=1.5,y=0.5,NoLabel,size=.1,layer=2]{E2}
\Vertex[x=1,y=0.3,NoLabel,size=.1,layer=2]{F2}
\Vertex[x=2.5,y=-.3,NoLabel,size=.1,layer=2]{G2}
\Edge[Direct](B2)(F2)
\Edge[Direct](F2)(B2)
\Edge[Direct](D2)(B2)
\Edge[Direct](B2)(E2)
\Edge[Direct](A2)(B2)
\Edge[Direct](B2)(A2)
\Edge[Direct](A2)(C2)
\Edge[Direct](A2)(D2)
\Edge[Direct](A2)(F2)
\Edge[Direct](C2)(F2)
\Edge[Direct](F2)(C2)
\Edge[Direct](C2)(G2)
\Edge[Direct](F2)(E2)
\Edge[Direct](E2)(G2)
\Edge[Direct](F2)(G2)
\Edge[Direct](D2)(E2)
\Edge[style=dashed](A1)(A2)
\Edge[style=dashed](B1)(B2)
\Edge[style=dashed](C1)(C2)
\Edge[style=dashed](D1)(D2)
\Edge[style=dashed](E1)(E2)
\Edge[style=dashed](F1)(F2)
\Edge[style=dashed](G1)(G2)
\Text[x=0,y=-1.5,layer=1]{Layer 1}
\Plane[x=-1,y=-1,width=4,height=3,layer = 1,NoBorder]{L1}
\Text[x=0,y=-1.5,layer=2]{Layer 2}
\Plane[x=-1,y=-1,width=4,height=3,color=green,layer=2,NoBorder]{L2}
\end{tikzpicture}
\caption{An illustration of observability and interconnectedness over two layers. Nodes connected by dashed lines across layers represent an agent. In Layer 1, observability is depicted, where directed edges indicate that agents can observe the actions of others, forming the observability graph. In Layer 2, agents interact and can affect each other's payoffs, as represented by the directed edges, forming the interaction graph.}
\label{fig:layers}
\end{figure}
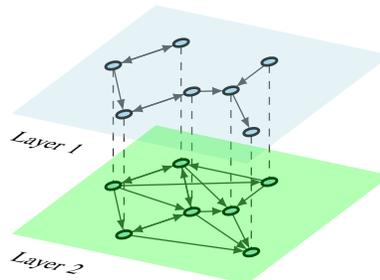

This paper is related to the broad literature of \textit{learning in games} studying whether simplistic learning or adaptation rules of self-interested agents can reach equilibrium, e.g., see the overview \cite{ref:Fudenberg09}. Such behavioral rules are generally based on the simplifying assumption that the opponents play according to some stationary strategies and the agents adapt their act either based on beliefs, as in the smoothed fictitious play \cite{ref:Fudenberg09}, or payoffs, as in the individual Q-learning \cite{ref:Leslie05}. The smoothed fictitious play is known to reach equilibrium in two-agent zero-sum and multi-agent potential games \cite{ref:Fudenberg09}. On the other hand, the individual Q-learning dynamics are known to reach equilibrium in two-agent zero-sum and \textit{two-agent} potential games \cite{ref:Leslie05}. Here, we show that they indeed also reach equilibrium in \textit{multi-agent} zero-sum and \textit{multi-agent} potential polymatrix games. 

Notably, the experience-weighted attraction (EWA) learning unifies the belief-based and payoff-based learning for normal-form games with some parameter controlling the weight of experiences (similar to the beliefs formed) and attractions (based on the payoffs received) in the actions taken \cite{ref:Camerer99}. However, agents still need to have access to opponent actions unless the attractions have the full weight, i.e., it is fully payoff-based. On the other hand, partial access to opponent actions can also lead to heterogeneous dynamics, as studied in \cite{ref:Arslantas23}. They only study two-agent zero-sum games where belief-based and payoff-based dynamics can play against each other.

In this paper, we demonstrate that the generalized individual Q-learning dynamics converge to quantal response equilibrium in multi-agent zero-sum and potential polymatrix games almost surely.  Furthermore, we quantify the significant improvement in convergence rate achieved through access to opponents' actions using numerical simulations.

The rest of the paper is organized as follows. We describe polymatrix games in Section \ref{sec:pre}. We present the new dynamics in Section \ref{sec:algo}. We provide analytical and numerical results, resp., in Sections \ref{sec:main} and \ref{sec:examples}. We conclude the paper with some remarks in Section \ref{sec:conclusion}. Two appendices include the proofs of propositions.

\section{Preliminary: Polymatrix Games}\label{sec:pre}

Consider $n$ agents having pairwise interactions with each other. We can represent their interactions via a graph (depicted in \ref{fig:layers}, Layer 2) $G=(I,E)$, where $I$ is the set of vertices and $E\subset I\times I$ is the set of directed edges, such that vertices and edges, resp., refer to agents and the pairwise interactions, as illustrated in Fig. \ref{fig:layers}. Then, we can model their interactions via a polymatrix game characterized by the tuple $\mathcal{P}=\langle G,(A^i,u^i)_{i\in I}^n\rangle$, where $A^i$ and $u^i:\prod_{j\in I}A^j\rightarrow \mathbb{R}$, resp., denote agent $i$'s \textit{finite} action set and payoff function. Furthermore, the payoff function $u^i(\cdot)$ can be written as
\be\label{eq:polypayoff}
u^i((a^t)_{t\in I}) = \sum_{j:(i,j)\in E} u^{ij}(a^i,a^j)\quad\forall (a^t\in A^t)_{t\in I},
\ee
for some sub-payoff function $u^{ij}:A^i\times A^j\rightarrow\mathbb{R}$ depending on the actions of the pair $(i,j)$ only. Indeed, the separability of the payoff function in \eqref{eq:polypayoff} is the defining characteristic of the polymatrix games \cite{ref:Howson72}. Without loss of generality, we let $u^{ij}(a^i,a^j)=0$ for all $(a^i,a^j)$ if $(i,j)\notin E$. Furthermore, we let each agent $i$ randomize their actions based on some mixed strategy $\pi^i\in \Delta^i$ independently, where $\Delta^i$ denotes the probability simplex over $A^i$. 

For example, we say that $\mathcal{P}$ is a \textit{zero-sum polymatrix game} if $\sum_{i\in I}u^i(a) = 0$ for all $a\in A:= \prod_{i\in I}A^i$. Correspondingly, we say that $\mathcal{P}$ is a \textit{potential polymatrix game} if there exists some potential function $\phi:A\rightarrow \mathbb{R}$ satisfying
\be \label{eq:potentialFunc}
    \phi(\tilde{a}^i,a^{-i}) - \phi(a^i,a^{-i}) = u^i(\tilde{a}^i,a^{-i}) - u^i(a^i,a^{-i}), 
\ee
for all $a^i, \tilde{a}^i \in A^i$, $a^{-i} \in A^{-i} \coloneqq \Pi_{j\neq i}A^j$, and $i \in I$.

We define agent $i$'s smoothed best response $\br^i:\mathbb{R}^{|A^i|}\rightarrow\Delta^i$ by
\be\label{eq:smoothed}
\br^i(q) := \argmax_{\mu \in \Delta^i}\left\{\mu^Tq + \tau H(\mu)\right\}\quad\forall q\in \mathbb{R}^{|A^i|},
\ee 
where $\tau>0$ is the temperature parameter controlling the exploration level and $H(\mu) = -\sum_{a}\mu(a)\log\mu(a)$ is the entropy regularization. There exists a unique smoothed best response due to the strict concavity of the regularization. Indeed, the smoothed best response $\br^i(q)$ can also be written as \cite{ref:Hofbauer05}
\be\label{eq:br}
\br^i(q)(a) = \frac{\exp(q(a)/\tau)}{\sum_{\tilde{a}\in A^i}\exp(q(\tilde{a})/\tau)}\quad\forall a\in A^i.
\ee
Note also that the smoothed best response has an inherent exploration property since every action have positive probability, i.e., $\br^i(q)(a)>0$ for all $a\in A^i$. 

Based on the smoothed best response, we focus on quantal response equilibrium (QRE) \cite{ref:McKelvey95} as a solution concept for $\mathcal{P}$. 
We say that a strategy profile $\pi = (\pi^i\in \Delta^i)_{i\in I}$ for $\mathcal{P}$ is QRE provided that 
\be
\pi^i = \br^i\left(u^i(\cdot,\pi^{-i})\right)\quad \forall i\in I,
\ee
where $\pi^{-i}=(\pi^j)_{j\neq i}$.\footnote{Given any function $f:A\rightarrow\mathbb{R}$, we let 
$f(\mu) = \E_{a\sim \mu}[f(a)]$ for any probability distribution $\mu$ over $A$, with a slight abuse of notation for simpler expressions.} 
Correspondingly, the QRE-gap of any strategy profile $\pi$ is defined by
\be \label{eq:qreGap}
\QRE(\pi) := \sum_{i\in I} \left|u^i(\pi^i,\pi^{-i}) -  u^i\left(\br^i(u^i(\cdot,\pi^{-i})),\pi^{-i}\right)\right|.
\ee

\section{Generalized Individual-Q Learning}\label{sec:algo}

We first describe the individual Q-learning dynamics \cite{ref:Leslie05} for the repeated play of \textit{two-agent} normal-form games where agents do not have access to opponent actions. Each agent $i$ assumes that the opponent $j\neq i$ plays according to some stationary mixed strategy $\pi^j\in \Delta^j$ in the repeated play of the underlying game. Then, the value of action $a^i\in A^i$, called the \textit{Q-function} and denoted by $\underline{q}^i:A^i\rightarrow\mathbb{R}$, is given by
\be\label{eq:q}
\underline{q}^i(a^i) = u^i(a^i,\pi^j)\quad\forall a^i\in A^i.
\ee
Since the agents do not observe the opponent actions, they estimate the Q-function based on the payoff received (different from the belief-based dynamics such as fictitious play). Let $q_k^i:A^i\rightarrow\mathbb{R}$ be agent $i$'s estimate at the $k$th repetition. Then, agent $i$ can update the estimate $q_k^i$ according to
\be\label{eq:qupdate}
q_{k+1}^i(a^i) = q_{k}^i(a^i) + \aalpha_k^i(a^i)\cdot\left(u_k^i - q_k^i(a^i)\right),
\ee
where $u_k^i = u^i(a_k^i,a_k^j)$ is the payoff received when the action profile $(a_k^i,a_k^j)$ is played and $\aalpha_k^i(a^i)\in [0,1]$ is some step size. 

Observe that agents only receive the payoff for the current action profile $(a_k^i,a_k^j)$. Therefore, we have $\aalpha_k^i(a^i) = 0$ for all $a^i\neq a_k^i$. Correspondingly, the estimates for the actions taken frequently get updated frequently. However, this leads to asynchronous update of the Q-functions for different actions. The individual Q-learning dynamics address this issue via the step sizes normalized by the probability of the current action gets played, i.e., $\br^i(q_k^i)(a^i)$, so that the Q-functions of every action get updated synchronously in the expectation \cite{ref:Leslie05}. Particularly, the step size $\aalpha_k^i$ is given by
\be\label{eq:stepsize}
\aalpha_k^i(a^i) = \mathbb{I}_{\{a^i=a_k^i\}} \min\left\{1,\frac{\alpha_k}{\br^i(q_k^i)(a^i)}\right\}\quad\forall a^i\in A^i,
\ee
where $\alpha_k\in(0,1)$ is some reference step size decaying to zero and $\mathbb{I}_{\{\cdot\}}$ is the indicator function. We further threshold the ratio $\alpha_k/\br^i(q_k^i)(a^i)$ from above by $1$. The threshold ensures that the step size $\aalpha_k^i(a^i)\in[0,1]$ for all $a^i$ so that the update \eqref{eq:qupdate} is the convex combination of the current estimate and the payoff received, and therefore, the estimates remain bounded as the payoffs are bounded if the initialization is from zero. 

Due to the threshold in \eqref{eq:stepsize}, the estimates for every action do not necessarily get updated synchronously in the expectation. However, the bounded estimates ensure that the denominator $\br^i(q_k^i)(a^i)$ is uniformly bounded from below by some $\epsilon>0$. Therefore, the decaying $\{a_k\}_{k\geq 0}$ ensures that eventually the threshold does not have any impact on the step size $\aalpha_k^i$, i.e., there exists some $K$ (not depending on the trajectory) such that $\alpha_k/\br^i(q_k^i)(a^i) <1$ for all $k>K$.

Next, we generalize the individual Q-learning to the repeated play of multi-agent polymatrix games where agents have limited access to opponent actions. To this end, we introduce the observability graph (depicted in Figure \ref{fig:layers}, Layer 1) $\uG=(I,\uE)$ where $\uE \subset I\times I$ is the set of directed edges such that $(i,j)\in \uE$ if agent $i$ can observe agent $j$'s actions, as illustrated in Fig. \ref{fig:layers}.

When agent $i$ has access to agent $j$'s action, agent $i$ can form a belief $\pi_k^j\in \Delta^j$ about agent $j$'s strategy (assumed to be stationary) based on the action history. For example, agent $i$ can take the weighted empirical average of the past actions, as in the widely studied fictitious play dynamics. Hence, we let agent $i$ use the reference step size $\alpha_k\in(0,1)$ in the update of the belief $\pi_k^j$ as
\be \label{eq:beliefUpdate}
\pi_{k+1}^j = \pi_k^j + \alpha_k\cdot(a_k^j - \pi_k^j)
\ee
for all $j:(i,j)\in \uE$, where we view the action $a_k^j$ as a degenerate mixed strategy giving probability one to the associated action.

Similar to \eqref{eq:q}, agent $i$'s Q-function is now given by
\begin{align}
\underline{q}^i(a^i) = u^i(a^i,\pi^{-i}) = \sum_{j\neq i} u^{ij}(a^i,\pi^j),
\end{align} 
where the second equality follows from \eqref{eq:polypayoff}. Then, based on the observability graph $\uG$, we can write the Q-function as
\begin{align}
\underline{q}^i(a^i)= &\,\sum_{j:(i,j)\in \uE} u^{ij}(a^i,\pi^j) + \sum_{j:(i,j)\notin \uE} u^{ij}(a^i,\pi^j).\label{eq:qpoly}
\end{align}
Agent $i$ can estimate the first summation on the right-hand side as 
\be
\cq_k^i(a^i) \coloneqq \sum_{j:(i,j)\in \uE} u^{ij}(a^i,\pi^j_k)
\ee
based on the beliefs $\pi_k^j$ formed for all $j$ such that $(i,j)\in \uE$. Indeed, \eqref{eq:beliefUpdate} yields that $\cq_k^i(\cdot)$ evolves according to
\be\label{eq:cq}
\cq_{k+1}^i(a^i) = \cq_{k}^i(a^i) + \alpha_k\cdot(\cu_k^i(a^i) - \cq_k^i(a^i))\quad\forall a^i\in A^i,
\ee
where 
\be
\cu_k^i(a^i) \coloneqq \sum_{j:(i,j)\in \uE} u^{ij}(a^i,a_k^j)\quad\forall a^i\in A^i.
\ee

On the other hand, agent $i$ does not form a belief about agents whose actions cannot be observed. Instead, agent $i$ can estimate the second summation based on the payoff $u^i_k = u^i(a_k^i,a_k^{-i})$ received as in the individual Q-learning dynamics \eqref{eq:qupdate}. However, the payoff received depends also on the pairwise interactions with the agents whose actions can be observed. Agent $i$ can still isolate the payoffs coming from the pairwise interactions with the agents whose actions cannot be observed by
\be\label{eq:hatu}
\widehat{u}^i_k := u^i_k - \cu_k^i(a_k^i)\in\mathbb{R}.
\ee
Then, the estimation for the second term on the right-hand side of \eqref{eq:qpoly}, denoted by $\hq_k^i:A^i\rightarrow\mathbb{R}$, can be updated according to
\be\label{eq:secondqupdate}
\hq^i_{k+1}(a^i) = \hq^i_k(a^i) + \aalpha_k^i(a^i) \cdot (\widehat{u}^i_k - \hq^i_k(a^i))\quad\forall\ a^i \in A^i,
\ee
where the step size $\aalpha_k^i(a^i)\in[0,1]$ is as described in \eqref{eq:stepsize}.

\begin{algorithm}[tb]
    \caption{Generalized Individual Q-learning Dynamics}
    \begin{algorithmic}\label{tab:algo}
    \STATE {\bfseries initialize:} the estimates $\cq_0^i,\hq^i_0\in\mathbb{R}^{A^i}$ arbitrarily
    \FOR{each stage $k=0,1,\ldots$}
    \STATE set $q_k^i = \cq_k^i + \hq^i_k$
    \STATE play $a^i_k \sim \br^i(q_k^i)$
    \STATE receive the payoff $u^i_k = u^i(a^i_k,a^{-i}_k)$ 
    \STATE observe the opponent actions $(a^j_k)_{j: (i,j) \in \uE}$ partially
    \STATE set $\cu_k^i(a^i) = \sum_{j: (i,j) \in \uE} u^{ij}(a^i,a^j_k)$ for all $a^i\in A^i$
    \STATE set $\widehat{u}^i_k = u^i_k - \cu_k^i(a_k^i)$
    \STATE update the estimates
    \begin{flalign*}
    &\cq^i_{k+1}(a^i) = \cq^i_k(a^i) + \alpha_k \cdot (\cu^i_k(a^i) - \cq^i_k(a^i))\quad\forall\ a^i \in A^i\\
    &\hq^i_{k+1}(a^i) = \hq^i_k(a^i) + \aalpha_k^i(a^i) \cdot (\widehat{u}^i_k - \hq^i_k(a^i))\quad\forall\ a^i \in A^i
    \end{flalign*}
    with the step sizes $\alpha_k\in(0,1)$ and 
    \[
    \aalpha_k^i(a^i) = \mathbb{I}_{\{a^i=a_k^i\}} \min\left\{1,\frac{\alpha_k}{\br^i(q_k^i)(a^i)}\right\}\quad\forall a^i\in A^i
    \]
    \ENDFOR
    \end{algorithmic}
\end{algorithm}

Given the estimates $\cq_k^i$ and $\hq_k^i$, the Q-function estimate is given by 
\be\label{eq:qestimate}
q_k^i(a^i) \coloneqq \cq_k^i(a^i) + \hq^i_k(a^i)\quad\forall a^i\in A^i.
\ee
Therefore, agent $i$ can play $a^i_k \sim \br^i(q_k^i)$, as in the individual Q-learning dynamics. 
Algorithm \ref{tab:algo} is a description of the generalized individual Q-learning dynamics for agent $i$. Note that the dynamics presented reduces to the smoothed fictitious play and individual Q-learning if $\uE=E$ or $\uE=\varnothing$, respectively. 

\begin{remark}\label{remark:algo}
We highlight the resemblance between \eqref{eq:cq} and \eqref{eq:secondqupdate}. They differ in the step sizes $\alpha_k\in[0,1]$ vs $\aalpha_k^i\in [0,1]^{A^i}$ and the targets $\cu_k^i\in \mathbb{R}^{A^i}$ vs $\widehat{u}_k^i\in \mathbb{R}$. The distinct target values arise from the need to account for counterfactual payoffs without directly observing opponents' actions. While agent $i$ knows what $\cu_k^i$ is for other possible actions, $\widehat{u}_k^i$ is only known for the current action taken by agent $i$. However, as shown later, both update have identical structures in the expectation conditioned on the play history due to the normalization in \eqref{eq:stepsize} (and for sufficiently large $k$ due to the threshold in \eqref{eq:stepsize}).
\end{remark}

\section{Convergence Results}\label{sec:main}

In this section, we characterize Algorithm \ref{tab:algo}'s convergence properties for any observability graph structure. Agents can have full, partial, or no access to observations. They can even have heterogeneous accesses, as in \cite{ref:Arslantas23} yet beyond two agents. 

We make the following assumption on the step sizes used.


\begin{assumption}\label{assm:stepsize}
The step size $\alpha_k \in (0,1)$ satisfies the following standard conditions: 
\be
\lim_{k\rightarrow\infty} \alpha_k = 0,\quad\sum_{k=0}^{\infty}\alpha_k = \infty,\quad\mbox{and}\quad \sum_{k=0}^{\infty}\alpha_k^2 < \infty
\ee 
 for stochastic approximation.
\end{assumption}

The following theorem shows that Algorithm \ref{tab:algo} reaches equilibrium in the sense that the weighted empirical averages of the actions taken, $\pi_k^i$'s, converges to QRE almost surely in multi-agent zero-sum and potential polymatrix games. Furthermore, the Q-function estimates, $q_k^i$'s, track the Q-functions associated with those empirical averages.

\begin{theorem}{Main Result}\label{thm:mainTheorem}
Consider a polymatrix game characterized by $\mathcal{P} = \langle G,(A^i,u^i)_{i\in I}\rangle$. Let every agent follow Algorithm \ref{tab:algo} and Assumption \ref{assm:stepsize} holds. If $\mathcal{P}$ is a zero-sum or potential polymatrix game, then the weighted empirical averages of the actions $(\pi^i_k)_{i \in I}$ evolving according to \eqref{eq:beliefUpdate} converge to QRE, i.e., 
\be\label{eq:resultQRE}
\QRE(\pi_k) \rightarrow 0
\ee
as $k\rightarrow \infty$ almost surely. Furthermore, the Q-function estimates are asymptotically belief-based in the sense that
\be\label{eq:resultqdiff}
\|q_k^i-u^i(\cdot,\pi_k^{-i})\|_2\rightarrow 0\quad\forall i\in I,
\ee
as $k\rightarrow \infty$ almost surely.\footnote{We can view finite dimensional functions (such as $q_k^i:A^i\rightarrow \mathbb{R}$) as vectors (in $\mathbb{R}^{|A^i|}$ since $A^i$ is a finite set).}     
\end{theorem}

\begin{proof}
The proof methodology is based on the stochastic approximation techniques \cite{ref:Benaim99}. The result follows from the stochastic approximation of Algorithm \ref{tab:algo} based on new Lyapunov function formulations. We first focus on approximating discrete-time updates in Algorithm \ref{tab:algo} via their limiting ordinary differential equations (o.d.e.). For example, we can write \eqref{eq:beliefUpdate} as
\be\label{eq:newbelief}
\pi_{k+1}^i = \pi_k^i + \alpha_k (\br^i(q_k^i) - \pi_k^i +\omega_k^i)\quad\forall i \in I,
\ee
where the stochastic approximation error is defined by
\be
\omega_k^i := a_k^i - \br^i(q_k^i).
\ee
Correspondingly, the update \eqref{eq:cq} can be written as
\be\label{eq:newcq}
\cq_{k+1}^i \equiv \cq_k^i + \alpha_k\left(\sum_{j: (i,j) \in \uE} u^{ij}(\cdot,\br^j(q_k^j)) - \cq_k^i + \widecheck{w}_k^i\right),
\ee
where the stochastic approximation error is defined by
\be
\widecheck{w}_k^i(a^i) \coloneqq \sum_{j: (i,j) \in \uE} (u^{ij}(\cdot,a_k^j) - u^{ij}(\cdot,\br^j(q_k^j)))\quad\forall a^i\in A^i.
\ee
Recall the discussion that the threshold in the step size $\aalpha_k^i(a^i)$ is ineffective, i.e., $\aalpha_k^i(a^i)<1$, for sufficiently large $k$. Therefore, for sufficiently large $k$, we can write \eqref{eq:secondqupdate} as
\be\label{eq:newq}
\hq^i_{k+1} \equiv \hq^i_k + \alpha_k\left(\sum_{j: (i,j) \notin \uE} u^{ij}(\cdot,\br^j(q_k^j))  - \hq^i_k + \widehat{w}_k^i\right),
\ee
where the stochastic approximation error is defined by
\begin{flalign}
\widehat{w}_k^i&(a^i) \coloneqq \frac{\mathbb{I}_{\left\{a^i=a_k^i\right\}}}{\br^i(q_k^i)(a_k^i)}\biggl(\widehat{u}^i_k - \hq_k^i(a^i)\biggr)\nn\\
&-\left(\sum_{j: (i,j) \notin \uE} u^{ij}(a^i,\br^j(q_k^j)) -\hq_k^i(a^i)\right)\quad\forall a^i\in A^i.
\end{flalign}

Let $x_k = [\cq_k^i;\hq_k^i; \pi_k^i]_{i \in I} \in \mathbb{R}^d$, where $d\coloneqq 3\prod_i |A^i|$, and $w_k = [\widecheck{w}_k^i;\widehat{w}_k^i ; \omega_k^i]_{i \in I}$. Given the filtration $\mathcal{F}_k \coloneqq \sigma(x_l,w_l; l\leq k)$, the stochastic approximation errors $(w_k)_{k\geq 0}$ form a Martingale difference sequence with $\mathrm{E}[\widecheck{w}^i_k\mid \mathcal{F}_k] \equiv 0$ and $\mathrm{E}[\widehat{w}^i_k\mid \mathcal{F}_k] \equiv 0$, where the expectation is taken with respect to the randomness induced by randomized actions. This yields that the updates \eqref{eq:newcq} and \eqref{eq:newq} have identical structures in the expectation conditioned on the history $\mathcal{F}_k$, as discussed in Remark \ref{remark:algo}.

To have a simpler expression for the evolution of $x_k$, we define the function $F(\cdot)$ by
\be\nn
F\left(\begin{bmatrix}\cq_k^i\\ \hq_k^i \\ \pi_k^i\end{bmatrix}_{i \in I}\right) = \begin{bmatrix}\sum_{j: (i,j) \in \uE} u^{ij}(\cdot,\br^j(\cq_k^j+\hq_k^j)) \\ \sum_{j: (i,j) \notin \uE} u^{ij}(\cdot,\br^j(\cq_k^j+\hq_k^j)) \\ \br^i(\cq_k^j+\hq_k^j)\end{bmatrix}_{i\in I}.
\ee
Since $\br^i(\cdot)$ is Lipschitz continuous \cite{ref:Gao17}, we can show that $F(\cdot)$ is also Lipschitz continuous.
Then, for sufficiently large $k$, we can write \eqref{eq:newbelief}, \eqref{eq:newcq}, and \eqref{eq:newq} as
\begin{flalign}\label{eq:discreteUpdate}
    x_{k+1} = x_k + \alpha_k(F(x_k) - x_k + w_k).
\end{flalign}
We can also show that the iterates $x_k$ remain bounded since $\alpha_k\in(0,1)$ and payoffs are bounded, as the updated values are simply the convex combination of previous $x_k$ and the bounded $F(x_k)$. Since the step sizes satisfy Assumption \ref{assm:stepsize}, we can apply stochastic approximation methods and obtain the limiting o.d.e. of \eqref{eq:discreteUpdate} as
$\dot{\pmb{x}} = F(\pmb{x}) - \pmb{x}$ \cite{ref:Benaim99}.
More explicitly, we have
\begin{subequations}\label{eq:ode}
\begin{flalign} 
&\frac{d\pcq^i}{dt}(\cdot) = \sum_{j:(i,j)\in \uE} u^{ij}(\cdot,\br^j(\pcq^{j}+\phq^j)) - \pcq^i(\cdot)\\
&\frac{d\phq^i}{dt}(\cdot) = \sum_{j:(i,j)\notin \uE} u^{ij}(\cdot,\br^j(\pcq^{j}+\phq^j)) - \phq^(\cdot)\\
&\frac{d\ppi^i}{dt} = \br^i(\pcq^i+\phq^i) - \ppi^i
\end{flalign}
\end{subequations}
for the continuous-time functions $\pcq^i:[0,\infty)\rightarrow\mathbb{R}^{A^i}$, $\phq^i:[0,\infty)\rightarrow\mathbb{R}^{A^i}$, and $\ppi^i:[0,\infty)\rightarrow\Delta^i$ for all $i\in I$.

The discrete-time iterates $x_k$ almost surely converges to a compact connected internally chain transitive set of the o.d.e. \eqref{eq:ode}. To characterize the limit behavior of \eqref{eq:discreteUpdate} through \eqref{eq:ode}, we follow the approach in \cite[Section 6.2]{ref:Benaim99} and call a continuously differentiable function $V:\mathbb{R}^d\rightarrow \mathbb{R}$ by a \textit{Lyapunov function} for some compact invariant set $\Lambda\subset \mathbb{R}^d$ provided that for any solution to \eqref{eq:ode}, we have
\begin{itemize}
\item $\dot{V}(\pmb{x}(t)) < 0$ for all $\pmb{x}(t)\notin \Lambda$
\item $\dot{V}(\pmb{x}(t)) = 0$ for all $\pmb{x}(t)\in \Lambda$.
\end{itemize}
Then, \cite[Proposition 6.4]{ref:Benaim99} says that every internally chain transitive set of \eqref{eq:ode} is contained in $\Lambda$ provided that $V(\Lambda)\coloneqq \{V(x):x\in\Lambda\}\subset \mathbb{R}$ has empty interior.

    \begin{figure*}[t]
    \centering
    \begin{subfigure}{.43\textwidth}
    \centering
    \includegraphics[width=\textwidth]{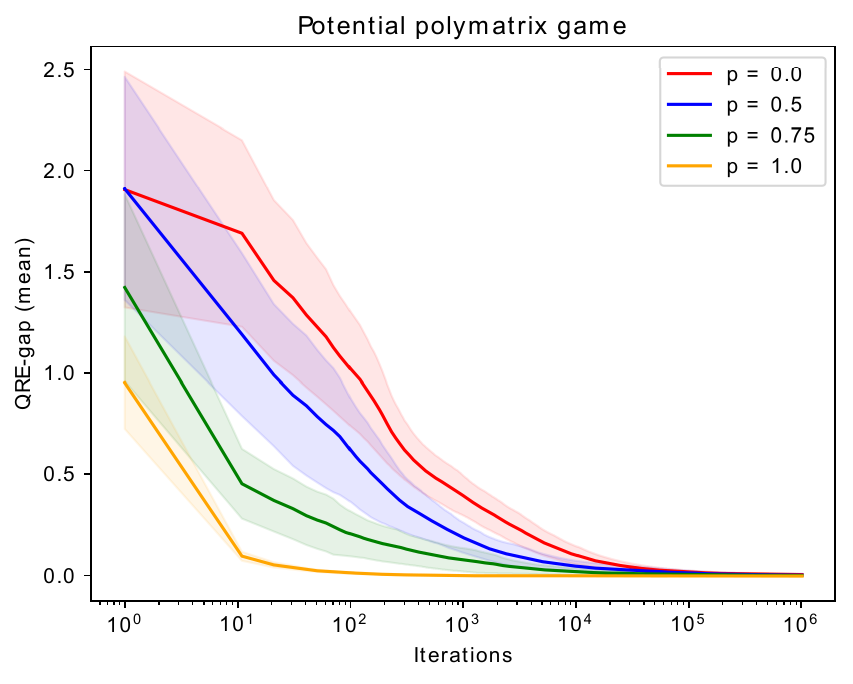}
    \caption{QRE-gap}\label{fig:ii_QRE}
    \end{subfigure}
    \begin{subfigure}{.43\textwidth}
    \centering
    \includegraphics[width=\textwidth]{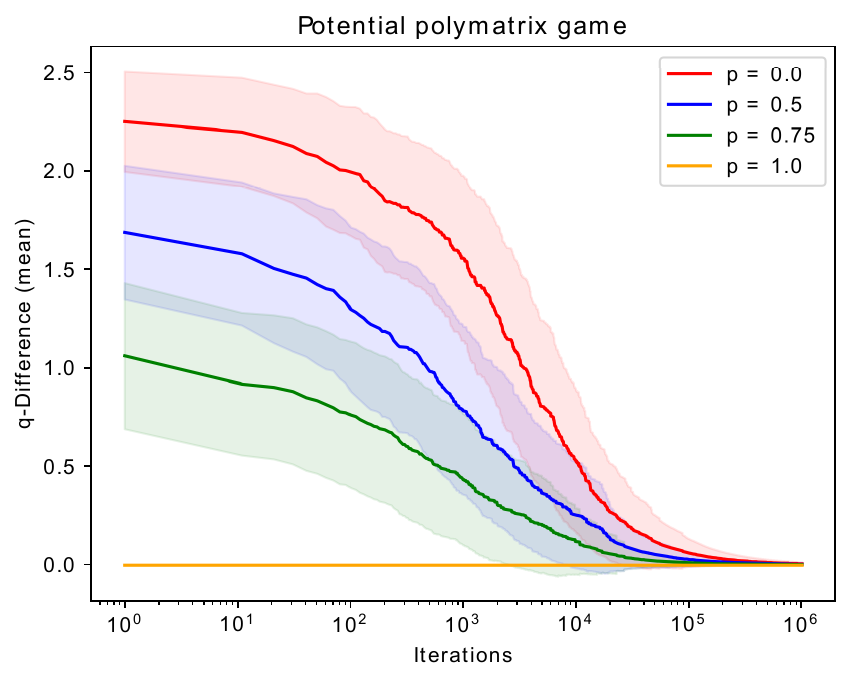}
    \caption{Q-difference}
    \label{fig:ii_Qdiff}
    \end{subfigure}
  \caption{Convergence of the QRE-gap (left), and $q_{\mathrm{diff}}$ (right) for different values of edge-connection probability $p$, in the potential polymatrix games. The solid curves represent the mean over $50$ independent trials and the shaded areas are $\pm 0.5$ standard deviations from the mean values.}\label{fig:potential}
\end{figure*}

Next, we formulate Lyapunov functions for \eqref{eq:ode}. We first focus on zero-sum polymatrix games and introduce the candidate $V_z(\cdot)$ defined by
\be\label{eq:V_z}
V_z(\cq,\hq,\pi) \coloneqq \sum_{i \in I} L^i(\cq^i+\hq^i,\pi),
\ee
where for each $i\in I$, the auxiliary $L^i(\cdot)$ is defined by
\begin{flalign}
    L^i(q^i,\pi) \coloneqq &\;\max_{\mu^i\in\Delta^i} \bigg\{q^i(\mu^i)  + \tau H(\mu^i)\bigg\} - \bigg(u^i(\pi) + \tau H(\pi^i)\bigg) \nn\\
    &+ \|q^i - u^i(\cdot,\pi^{-i})\|_2.\label{eq:Li}
\end{flalign}
 
\begin{remark}
The Lyapunov function $V_z(\cdot)$ reduces to the one presented in \cite{ref:Park24} (and \cite{ref:Hofbauer05}) if agents have full access to observations (and there are only two agents). 
\end{remark}
 
The following proposition shows that the candidate $V_z(\cdot)$ is a Lyapunov function for \eqref{eq:ode}, and \eqref{eq:resultQRE} and \eqref{eq:resultqdiff} hold for zero-sum polymatrix games.

\begin{figure*}[t]
    \centering
    \begin{subfigure}{.43\textwidth}
    \centering
    \includegraphics[width=\linewidth]{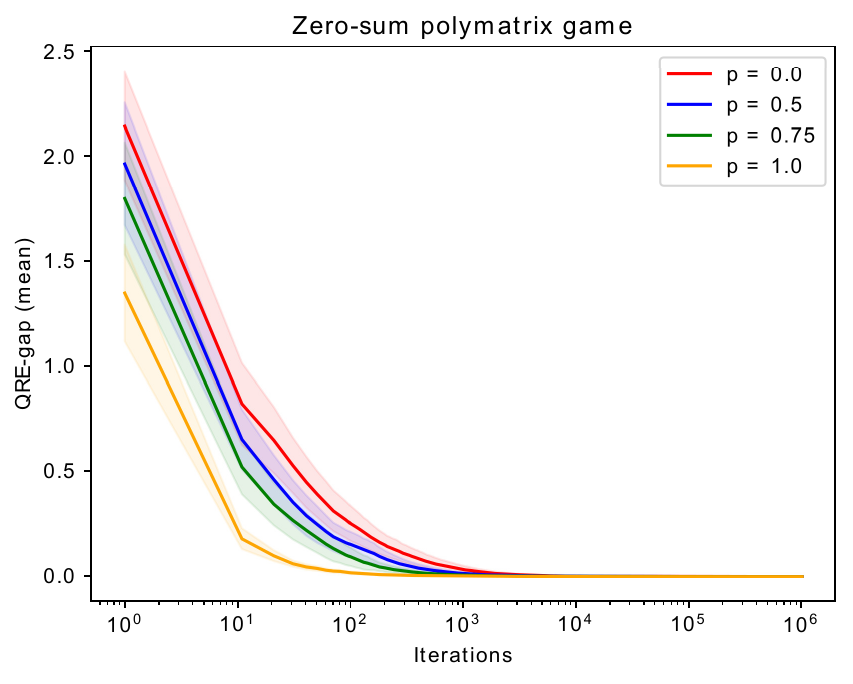}
    \caption{QRE-gap}\label{fig:zs_QRE}
    \end{subfigure}
    \begin{subfigure}{.43\textwidth}
    \centering
    \includegraphics[width=\linewidth]{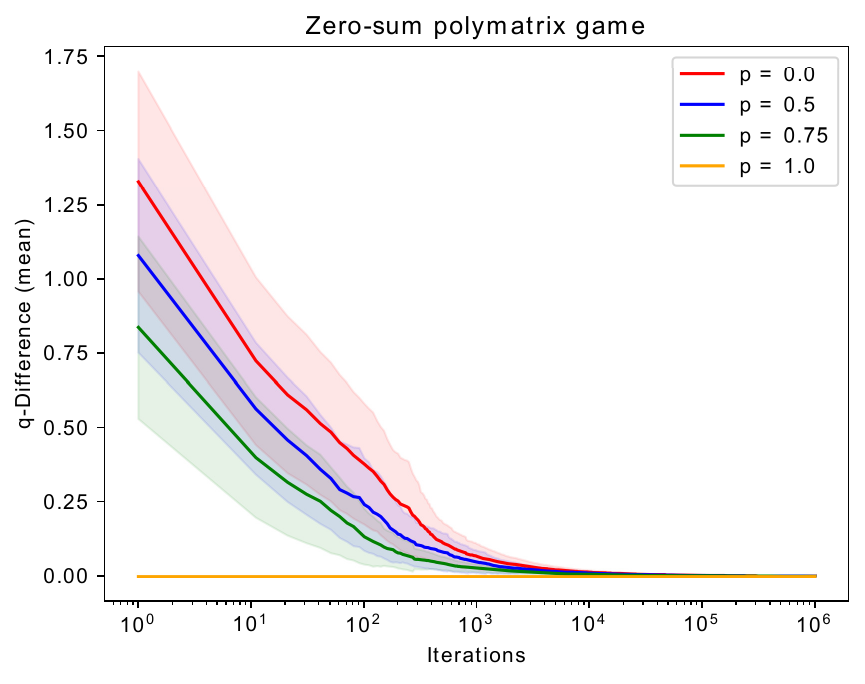}
    \caption{Q-difference}\label{fig:zs_Qdiff}
    \end{subfigure}
     \caption{Convergence of the QRE-gap (left), and $q_{\mathrm{diff}}$ (right) for different values of edge-connection probability $p$, in the zero-sum polymatrix games. The solid curves represent the mean over $50$ independent trials and the shaded areas are $\pm 0.5$ standard deviations from the mean values.}\label{fig:zerosum}
\end{figure*}

\begin{proposition}\label{prop:zerosum}
     Given a zero-sum polymatrix game $\mathcal{P}$ and the o.d.e. \eqref{eq:ode}, the function $V_z(\cdot)$ is a Lyapunov function for
     \be\nn
     \Lambda:= \{(\cq,\hq,\pi): \pi^i = \br^i(\cq^i+\hq^i)\mbox{ and } \cq^i+\hq^i \equiv u^i(\cdot,\pi^{-i})\}
     \ee
     and $V_z(\Lambda) = 0$.
     \end{proposition}

We next focus on potential polymatrix games with the potential function $\phi(\cdot): A \rightarrow \mathbb{R}$ and introduce the candidate $V_p(\cdot)$ defined by 
\begin{flalign}
        V_p(\cq,\hq,\pi) \coloneqq &\;-\phi(\pi) -\sum_{i \in I} \tau H(\pi^i) \nn\\
        &+ 2\sum_{i \in I} \|q^i - u^i(\cdot,\pi^{-i})\|_2.\label{eq:V_p}
    \end{flalign}
    
\begin{remark}
The Lyapunov function $V_p(\cdot)$ reduces to the one in \cite{ref:Hofbauer05} if agents have full access to observations.
\end{remark}

The following proposition shows that the candidate $V_p(\cdot)$ is a Lyapunov function for \eqref{eq:ode} in potential polymatrix games.
   
    \begin{proposition} \label{prop:V_p}
        Given a potential polymatrix game $\mathcal{P}$ and the o.d.e. \eqref{eq:ode}, the function $V_p(\cdot)$ is a Lyapunov function for $\Lambda$, which is defined in Proposition \ref{prop:zerosum} yet $V_p(\Lambda)$ is not necessarily singleton.
    \end{proposition}
    
    Note that $V_p$ is continuously differentiable. Therefore, Sard's Lemma yields that $V_p(\Lambda)$ has empty interior, e.g., see \cite[Corollary 3.28]{ref:Benaim05}. This implies that \eqref{eq:resultQRE} and \eqref{eq:resultqdiff} also hold for potential polymatrix games. This completes the proof of Theorem \ref{thm:mainTheorem}.
\end{proof}

\section{Illustrative Examples}\label{sec:examples}

In this section, we present simulation results. We generate fully-connected polymatrix games with uniformly distributed edge-payoffs $u^{ij}$ for each edge in ${(i,j): i<j}$. For zero-sum games, we set $u^{ji} = -u^{ij}$, and for potential games, $u^{ji} = u^{ij}$. In zero-sum games, payoffs sum to zero, $\sum_{i \in I} u^i(a) = 0$, while in potential games, $\phi(a) = \frac{1}{2}\sum_{i \in I} u^i(a)$.

We focus on $(n=4)$-agent games where each agent has $3$ actions, i.e., $|A^i|=3$ for all $i$. We set the temperature parameter $\tau = 0.25$.
We generate the observability graphs based on the Erd\H{o}s R\'{e}nyi model with probabilities $p \in \{0,0.5,0.75,1\}$ such that $p=0$ and $p=1$ correspond, resp., to the no and full access to observations. We run $50$ independent trials for each $p$ and at each trial, we generate a random observability graph according to this model. We also run our simulations over $10^6$ repetitions of the underlying game. 

To measure the performance of Algorithm \ref{tab:algo}, we use two metrics: QRE-gap, as described in \eqref{eq:qreGap}, and 
\begin{flalign}
    q_{\mathrm{diff}}(q_k,\pi_k) \coloneqq \frac{1}{n} \sum_{i \in I} \|q_k^i(\cdot) - u^i(\cdot,\pi^{-i}_k)\|_2
\end{flalign}
measuring the tracking error between $q_k^i$ and $u^i(\cdot,\pi_k^{-i})$ averaged across agents, similar to \eqref{eq:resultqdiff}. We plot the evolutions of these metrics for potential polymatrix and zero-sum polymatrix games, resp., in Figs. \ref{fig:potential} and \ref{fig:zerosum}. We observe that these metrics decay to zero, as suggested by Theorem \ref{thm:mainTheorem}. 

\begin{figure}[t!]
\centering
\includegraphics[width=2.0in]{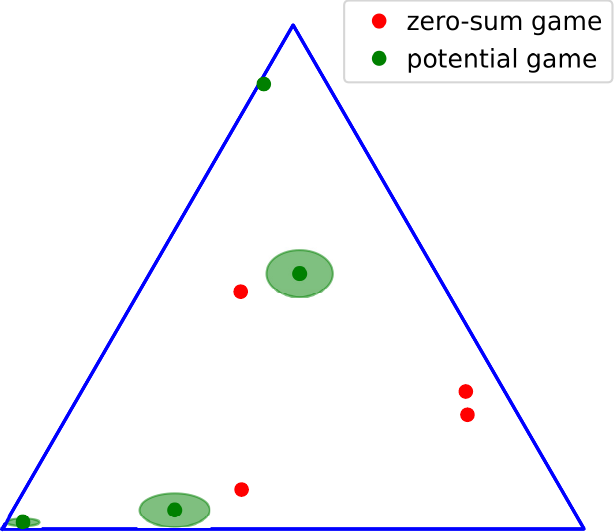}
\caption{An illustration of $\pi^i_k$'s over the probability simplex $\Delta^i\subset\mathbb{R}^3$ for each agent $i$ at the end of $10^6$ iterations and over $50$ independent trials for $p=1$. Green and red dots are, resp., for the potential and zero-sum games. The shaded areas represent the standard deviations from the mean values.}\label{fig:empirical}
\end{figure}
We achieve asymptotic convergence for any edge-connection probability \(p\), with the rate improving as $p$ increases. Convergence is slowest with no observations \( p = 0 \), and fastest with full observations \( p = 1 \). We observe that even though the QRE-gap decays at similar rates for the potential and zero-sum games under full observations, i.e., $p=1$, the QRE-gap decays slower for the potential games than the zero-sum games under partial/no observations, i.e., $p<1$. Indeed, the $q_{\mathrm{diff}}$ metrics also decay to zero slower for the potential games and this can slow down the QRE-gap's decay under partial/no observations. Exploration plays an important role in minimizing the tracking error. The equilibrium for the zero-sum game turns out to be more mixed than the potential game, as illustrated in Fig. \ref{fig:empirical}. 


\section{Conclusion}\label{sec:conclusion}

We introduced the generalized individual Q-learning dynamics to address the challenge of limited observations in non-cooperative multi-agent systems. By combining belief-based and payoff-based learning, this approach demonstrates faster convergence to quantal response equilibrium in networked interconnections of more than two self-interested agents. The proposed dynamics offer a significant improvement over existing methods, establishing theoretical convergence guarantees in multi-agent zero-sum and potential polymatrix games. Through numerical simulations, we validate the effectiveness of our approach and experimentally observe the benefits of observing opponents' actions. 

We can list some of the future research directions as follows: $(i)$ generalizing other payoff-based dynamics to partial observation settings, $(ii)$ addressing limited access issue for learning in stochastic games, and $(iii)$ quantifying the impact of the observations on the learning rate analytically.

\appendices
\section{Proof of Proposition \ref{prop:zerosum}}

We first show the non-negativity of each $L^i \geq 0$, and therefore, the non-negativity of $V_z\geq 0$. To this end, we let $\widehat{\mu} \coloneqq \argmax_{\mu\in \Delta^i} \{u^{i}(\mu,\pi^{-i})+\tau H(\mu)\}$, which is unique due to the entropy regularization. Then, \eqref{eq:Li} yields that we can bound $L^i(\cdot)$ from below by
    \begin{flalign}
        L^i(q^i,\pi) &\stackrel{(a)}{\geq} q^{i}(\widehat{\mu}) - u^{i}(\widehat{\mu},\pi^{-i})+ \|q^i - u^i(\cdot,\pi^{-i})\|_2\nn\\
        & \stackrel{(b)}{\geq} 0, \label{eq:Lbound}
    \end{flalign}
    where $(a)$ is due to the definitions of $L^i$ and $\widehat{\mu}$, and $(b)$ follows from $\|x\|_2 \geq \|x\|_\infty \geq |\widehat{\mu}^Tx|$ for all $x \in \mathbb{R}^{\left|A^i\right|}$ and $\widehat{\mu}\in \Delta^i$. Note that $(a)$ is an equality if only if $\widehat{\mu}=\br^i(q^i)$ and $\widehat{\mu} = \pi^i$. On the other hand, $(b)$  is an equality if only if $q^i - u^{i}(\cdot,\pi^{-i}) \equiv 0$ since $\widehat{\mu}$ is a fully mixed strategy due to the entropy regularization. 
    Therefore, we have $L^i(q^i,\pi) = 0$ if and only if $\pi^i = \br^i(q^i)$, and $q^i = u^i(\cdot, \pi^{-i})$. 
    
   Next, we focus on the time derivative of $V_z$. Since the underlying game is zero sum, we have 
   \be\label{eq:zerosum}
   \sum_{i\in I} u^i(\pi) = 0\quad\forall \pi \in \prod_{i \in I} \Delta^i.
   \ee 
This implies that we can write $V_z(\cdot)$ as
    \begin{flalign}
    V_z(\cq,\hq,\pi) = \sum_{i\in I} \bigg(&\max_{\mu\in\Delta^i}\{q^{i}(\mu)+\tau H(\mu)\} - \tau H(\pi^i) \nn\\
    &+ \|q^i - u^i(\cdot,\pi^{-i})\|_2\bigg).\label{eq:sum}
    \end{flalign}
Due to the smoothness of the entropy regularization, the envelope theorem yields that the time derivative of the first term in the summation \eqref{eq:sum} is given by
\begin{flalign}
\frac{d}{dt}\max_{\mu\in\Delta^i}\{\pq^{i}(\mu)&+\tau H(\mu)\} = \frac{d\pq^i}{dt}(\br^i(\pq^i)) \nn\\
&= \sum_{j\in I} u^{ij}(\br^i(\pq^i),\br^j(\pq^j)) - \pq^i(\br^i(\pq^i)),\label{eq:env}
\end{flalign}
where the second equality follows from \eqref{eq:ode} as $\pq \equiv \pcq+\phq$.
By \eqref{eq:zerosum} and \eqref{eq:env}, we can write the time derivative of $V$ as
    \begin{flalign}
        &\frac{d}{dt}V_z(\pcq,\phq,\ppi) = -\sum_{i \in I} \Big(\pq^i(\br^i(\pq^i)) + \|\pq^i - u^i(\cdot,\ppi^{-i})\|_2 \nn \\
        &\hspace{.6in}+ \tau \nabla H(\ppi^i)(\br^i(\pq^i)-\ppi^i) \Big).
        \end{flalign}
        By \eqref{eq:Li} and \eqref{eq:zerosum}, adding and subtracting $\tau H(\br^i(\pq^i))-\tau H(\ppi^i)$ for each $i$ yields that 
        \begin{flalign}
        \frac{d}{dt}V_z(\pcq,\phq,\ppi) &= -\sum_{i\in I} L^i(\pq^i,\ppi) \nn\\
        &\hspace{-.4in}+\tau\sum_{i\in I} (H(\br^i(\pq^i))-H(\ppi^i) - \nabla H(\ppi^i)(\br^i(\pq^i)-\ppi^i)) \nn\\
        &\leq \sum_{i \in I} -L^i(\phq^i,\ppi),\label{eq:final}
    \end{flalign}
    where the inequality follows from the concavity of the entropy, i.e., $H(y)-H(x) - \nabla H(x)(y-x) \leq 0$ for all $x$, $y$. 
    
    By \eqref{eq:Lbound} and \eqref{eq:final}, we have $\dot{V}_z\leq 0$ with equality if and only if $L^i = 0$ for all $i \in I$ and $H(\br^i(\pq^i))-H(\ppi^i) = \nabla H(\ppi^i)(\br^i(\pq^i)-\ppi^i)$. Therefore, we have $\dot{V}_z(\pcq,\phq,\ppi) = 0$ if and only if $\ppi^i = \br^i(\pq^i)$ and $\pq^i = u^i(\cdot,\ppi^{-i})$. 
    
    \section{Proof of Proposition \ref{prop:V_p}}
    By the chain rule, the time derivative of the first term on the right-hand side of \eqref{eq:V_p} is given by
    \begin{flalign}
    \frac{d}{dt}\sum_{a\in A} \phi(a)\prod_{i\in I}&\ppi^i(a^i) 
    =\sum_{i\in I} \phi(\br^i(\pq^i),\ppi^{-i}) - \phi(\ppi^i,\ppi^{-i})\nn\\
    &=\sum_{i\in I} u^i(\br^i(\pq^i),\ppi^{-i}) - u^i(\ppi^i,\ppi^{-i}),
    \end{flalign}
    where the last equality is due to \eqref{eq:potentialFunc}. Therefore, the time derivative of $V_p$ is given by
    \begin{flalign}
    \frac{dV_p}{dt} &= - \sum_{i\in I} \big(u^i(\br^i(\pq^i),\ppi^{-i}) - u^i(\ppi^i,\ppi^{-i}) \nn\\
    &+ \tau \nabla H(\ppi^i)\cdot (\br^i(\pq^i)-\ppi^i) + 2\|\pq^i-u^i(\cdot,\ppi^{-i})\|_2\big).
    \end{flalign}
    Then, adding the non-negative $\pq^i(\br^i(\pq^i))+ \tau H(\br^i(\pq^i)) - \pq^i(\ppi^i)-\tau H(\ppi^i) \geq 0$ to the right-hand side yields that
    \begin{flalign}
        \frac{dV_p}{dt} \leq& \sum_{i \in I} \tau \big(H(\br^i(\pq^i))-H(\ppi^i) - \nabla H(\ppi^i)(\br^i(\pq^i)-\ppi^i) \nn\\
        &+u^i(\ppi^i,\ppi^{-i}) - u^i(\br^i(\pq^i),\ppi^{-i}) +\pq^i(\br^i(\pq^i)) - \pq^i(\ppi^i) \nn\\
        &-2\|\pq^i-u^i(\cdot,\ppi^{-i})\|_2\big).
    \end{flalign}
       As in \eqref{eq:final}, by the concavity of the entropy function, we have
    \begin{flalign} \label{eq:VpEntropyCancel}
        \frac{dV_p}{dt} \leq& \sum_{i \in I} \big(u^i(\ppi^i,\ppi^{-i}) - u^i(\br^i(\pq^i),\ppi^{-i}) \nn \\
        &+\pq^i(\br^i(\pq^i)) - \pq^i(\ppi^i)  -2\|\pq^i-u^i(\cdot,\ppi^{-i})\|_2\big).
    \end{flalign}
    Viewing $q^i\in \mathbb{R}^{A^i}$ as a $|A^i|$-dimensional vector, we can write $q^i(\mu) = (\mu^T)q^i(\cdot)$ for any $\mu \in \Delta^i$. Therefore, we obtain
    \begin{flalign} \label{eq:VpFinal}
        \frac{dV_p}{dt} \leq& \sum_{i \in I} \big((\ppi^i - \br^i(\pq^i))^T u^i(\cdot,\ppi^{-i}) \nn \\
        &- (\ppi^i - \br^i(\pq^i))^T \pq^i(\cdot) -2\|\pq^i-u^i(\cdot,\ppi^{-i})\|_2\big) \nn\\
        =&\sum_{i \in I} \big((\br^i(\pq^i) - (\ppi^i))^T (\pq^i(\cdot) - u^i(\cdot,\ppi^{-i})) \nn \\
        &-2\|\pq^i-u^i(\cdot,\ppi^{-i})\|_2\big) \leq 0,
    \end{flalign}
    where the last inequality follows from a similar argument that is used in \eqref{eq:Lbound}, and we use the fact that $2\|x\|_2 \geq 2\|x\|_\infty \geq |(\ppi^i-\br^i(\pq^i))^T(x)|$ for all $x \in \mathbb{R}^{\left|A^i\right|}$, because $\ppi^i \in \Delta^i$, and $\br^i(\pq^i) \in \Delta^i$. Similar to \eqref{eq:Lbound}, the equality condition can only be satisfied if $\br^i(\pq^i) = \ppi^i$, and $\pq^i = u^i(\cdot,\ppi^{-i})$.
\bibliographystyle{IEEEtran}
\bibliography{mybib}

\end{document}